\documentclass[lettersize,journal]{IEEEtran}
\usepackage{amsmath}
\usepackage{amsfonts}
\usepackage{algorithmic}
\usepackage{algorithm}
\usepackage{array}
\usepackage[caption=false,font=footnotesize,labelformat=parens,labelfont=sf,textfont=sf]{subfig}
\usepackage{textcomp}
\usepackage{stfloats}
\usepackage{url}
\usepackage{verbatim}
\usepackage{graphicx}
\usepackage{cite}
\usepackage{color}
\usepackage{bm}
\usepackage{gensymb}
\usepackage{caption}
\usepackage[font=footnotesize]{caption} 
\usepackage{xcolor}
\usepackage{booktabs} 
\usepackage{makecell}
\usepackage{amsthm}

\newtheorem{lemma}{Lemma}

\usepackage{titlesec}
\titlespacing{\section}{0pt}{*0.5}{*0.4}
\titlespacing{\subsection}{0pt}{*0.3}{*0.3}
\titlespacing{\subsubsection}{0pt}{*0.3}{*0.2}

\begin{document}
\title{\huge Hybrid Quantum Convolutional Neural Network-Aided\\ Pilot Assignment in Cell-Free Massive MIMO Systems\vspace{-2mm}}

\author{Doan Hieu Nguyen\IEEEauthorrefmark{1}, 
        Xuan Tung Nguyen\IEEEauthorrefmark{1}, Seon-Geun Jeong,
        Trinh Van Chien, Lajos Hanzo, ~\IEEEmembership{Life Fellow,~IEEE},
        Won Joo Hwang, ~\IEEEmembership{Senior Member,~IEEE} \vspace{-1cm}%

\thanks{This work was supported in part by the Quantum Computing based on Quantum Advantage challenge research(RS-2024-00408613) through the National Research Foundation of Korea(NRF) funded by the Korean government (Ministry of Science and ICT(MSIT)); in part by the National Research Foundation of Korea(NRF) grant funded by the Korea government(MSIT)(RS-2024-00336962); in part by Institute of Information $\&$ communications Technology Planning $\&$ Evaluation(IITP) under the Artificial Intelligence Convergence Innovation Human Resources Development(IITP-2025-RS-2023-00254177) grant funded by the Korea government(MSIT); and in part by the IITP(Institute of Information $\&$ Communications Technology Planning $\&$ Evaluation)-ITRC(Information Technology Research Center) grant funded by the Korea government(Ministry of Science and ICT)(IITP-2025-RS-2023-00260098).}

\thanks{Doan Hieu Nguyen, and Seon-Geun Jeong are with the Department of Information Convergence Engineering, Pusan National University, South Korea (e-mail: hieu.nguyendoan@pusan.ac.kr, wjdtjsrms11@pusan.ac.kr). Nguyen Xuan Tung is with Faculty of Interdisciplinary Digital Technology, PHENIKAA University, Yen Nghia, Hanoi, Viet Nam (e-mail: tung.nguyenxuan@phenikaa-uni.edu.vn). Trinh Van Chien is with the School of Information and Communications Technology, Hanoi University of Science and Technology, Vietnam (e-mail: chientv@soict.hust.edu.vn). Lajos Hanzo is with the Department of Electronics and Computer Science, University of Southampton, U.K. (e-mail: lh@ecs.soton.ac.uk). Won-Joo Hwang (corresponding author) is with  School of Computer Science and Engineering, Digital-X AIoT Research Center, School of Computer Science and Engineering, Pusan National University, South Korea (e-mail: wjhwang@pusan.ac.kr).}

\thanks{\IEEEauthorrefmark{1}Equal contribution.}

\thanks{Personal use of this material is permitted. However, permission to use this material for any other purposes must be obtained from the IEEE by sending a request to pubs-permissions@ieee.org. }}



\maketitle
\begin{abstract}
A sophisticated hybrid quantum convolutional neural network (HQCNN) is conceived for handling the pilot assignment task in cell-free massive MIMO systems, while maximizing the total ergodic sum throughput. The existing model-based solutions found in the literature are inefficient and/or computationally demanding. Similarly, conventional deep neural networks may struggle in the face of high-dimensional inputs, require complex architectures, and their convergence is slow due to training numerous hyperparameters. The proposed HQCNN leverages parameterized quantum circuits (PQCs) relying on superposition for enhanced feature extraction. Specifically, we exploit the same PQC across all the convolutional layers for customizing the neural network and for accelerating the convergence. Our numerical results demonstrate that the proposed HQCNN offers a total network throughput close to that of the excessive-complexity exhaustive search and outperforms the state-of-the-art benchmarks. 
\vspace{-5mm}
\end{abstract}

\section{Introduction}
Cell-free massive multiple-input multiple-output (MIMO) leads to an innovative wireless communication architecture eliminating cell-boundaries by allowing multiple access points (APs) to jointly serve multiple users over a wide area \cite{tung2024distributed}.  It improves both the spatial diversity, and geographic load-balancing, plus boosts the coverage, throughput, and fairness \cite{yu2023learning}. The users transmit the pilot signals in the uplink to their APs for channel estimation. However, it would be extremely wasteful to have long user-specific pilots, hence they are reused by the co-channel APs, which leads to interference also known as pilot contamination. This phenomenon reduces the quality of channel estimation and degrades the spectral efficiency \cite{wang2024pilot}. Hence, the pilot assignment schemes must be carefully designed for mitigating pilot contamination. We emphasize that the pilot assignment constitutes a challenging combinatorial problem, making the optimal full search computationally prohibitive. In \cite{ngo2017cell}, the authors proposed a greedy pilot assignment scheme for maximizing the minimum data rate by mitigating the pilot contamination. Zaher \textit{et al.} \cite{zaher2022learning} introduced a so-called master-AP pilot assignment scheme, where each user identifies its AP having the highest channel gain as the master AP. The pilot signals are then assigned by minimizing the mutual interference at the master APs. A location-based pilot allocation strategy was introduced in \cite{nguyen2023efficient}, which divides the coverage area into grid-based regions and assigns the pilot signals into disjoint subsets. Explicitly, each region is assigned a subset, so that users sharing the same pilot are geographically distant, hence minimizing interference across the network. However, these heuristic methods focus on interference reduction, rather than directly maximizing the total ergodic throughput. In \cite{kim2020deep}, a convolutional neural network (CNN) was shown to perform well in the testing phase by utilizing an exhaustive search to gather data for training labels. However, utilizing large number of trainable parameters makes CNN becomes burdensome for large-scale systems including multiple users and APs.

Deep neural networks (DNN) having numerous hyperparameters usually require a long time for convergence, hence necessitating a new design for supporting large-scale communication systems. Quantum Machine Learning (QML), leveraging Parameterized Quantum Circuits (PQCs) as computational layers in classical models, presents a promising alternative to former approaches due to its unique quantum advantages. These advantages arise from the fundamental properties of qubits and quantum circuits: (1) Superposition, allowing qubits to exist in both 0 and 1 states simultaneously \cite{fan2023hybrid}; (2) Tensor product structure, allowing qubits to encode exponentially more states than classical bits \cite{du2025quantum}; and (3) Unitary evolution, preserving superposition throughout computation in quantum circuits. By exploiting these properties, QML models can harness the superposition to simultaneously process multiple states, theoretically reducing the runtime below those of classical counterparts \cite{fan2023hybrid}. Moreover, quantum states are well-suited for handling large-scale datasets, despite utilizing a reduced number of hyperparameters \cite{hanzo2025quantum}. \textcolor{black}{Recently, QML has attracted substantial research attention, especially in the classification tasks \cite{hur2022quantum, jeong2023hybrid, fan2023hybrid, oh2020tutorial}. Especially, QML has demonstrated superior performance over classical CNNs in multi-label image classification tasks \cite{fan2023hybrid, oh2020tutorial}. However, prior studies often overlook critical factors such as the number of trainable parameters in both classical and quantum components, as well as the influence of noise in quantum environments \cite{oh2020tutorial}. These limitations necessitate a carefully designed hybrid quantum-CNN model that accounts for both model complexity and quantum noise.}




These finding inspired us to propose an HQCNN for handling with the pilot assignment task of cell-free massive MIMO communication, with the objective of maximizing the total network throughput. In contrast to the layer-specific quantum circuit designs used in \cite{hur2022quantum}, we propose to use the same parameterized quantum circuit across all the convolutional layers. This allow us to significantly reduce the number of training parameters and accelerates the convergence during the training phase. The proposed HQCNN is trained using both supervised and unsupervised techniques to analyze models performance from different perspectives. To the best of our knowledge, this is the first study harnessing QML for solving the pilot assignment problem of cell-free massive MIMO communication as a multi-class classification task. Our numerical results demonstrate that the proposed HQCNN converges faster than classical deep learning with a margin of $2\%$ of the global optimum. We also estimate the impacts of noisy environment and scalability of HQCNN in this work.

\section{System Model and Problem Formulation}
\label{sec:SysMode-ProFor}

\subsection{System Model}
We consider a cell-free massive MIMO system operating in the time division duplexing model. Expliciting, $M$ APs - each equipped with $L$ antennas - jointly serve $K$ single-antenna users randomly distributed in the coverage area. In each coherence interval, the system assigns $\tau_p$ orthogonal pilot signals for estimating the propagation channels in the uplink training phase. The binary variable $p_{kt} \in \{ 0,1 \}$ indicates the pilot assignment:  $p_{kt}=1$ if user~$k$ occupies the $t$-th pilot $\pmb{\varphi}_t \in \mathbb{C}^{\tau_p}$, where we have $\| \pmb{\varphi}_t \|^2 =1$. Otherwise, $p_{kt}=0$. Let us denote the channel between AP~$m$ and user~$k$ by $\mathbf{g}_{mk} \in \mathbb{C}^L$. All the elements of $\mathbf{g}_{mk}$ are independent and identically distributed (i.i.d) as $\mathcal{C}\mathcal{N}(0,\beta_{mk})$, with $\beta_{mk}$ representing the large-scale fading (LSF).

\subsubsection{Uplink Pilot Training} 
All the users simultaneously transmit the assigned pilot signals to the APs. The signal received at AP~$m$ is $\mathbf{Y}_m = \sqrt{\tau_p \rho_p} \sum_{k=1}^K\sum\nolimits_{t=1}^{\tau_p}p_{kt} \mathbf{g}_{mk} \bm{\varphi}_t^H + \mathbf{W}_{m}$, where $\rho_p$ is the normalized signal-to-noise ratio (SNR) of each pilot symbol. Furthermore, $\mathbf{W}_{m}\in \mathbb{C}^{L\times \tau_p}$ represents the additive white Gaussian noise (AWGN) whose elements are i.i.d $ \mathcal{C}\mathcal{N}(0,1)$. Similar to \cite{tung2024distributed}, the channel estimate $\hat{\mathbf{g}}_{mk}$ is defined by  exploiting the linear minimum mean square error (LMMSE) estimator at the core network as
\begin{multline}
\label{eq:channel estimation short}
    \hat{\mathbf{g}}_{mk} =  \sqrt{\tau_p \rho_p} c_{mk} \sum\nolimits_{j=1}^K\sum\nolimits_{t'=1}^{\tau_p}\sum\nolimits_{t=1}^{\tau_p}\mathbf{g}_{mk}p_{jt'}\pmb{\varphi}_{t'}^Hp_{kt}\pmb{\varphi}_t \\ + \mathbf{W}_{m} c_{mk} \sum\nolimits_{t=1}^{\tau_p}p_{kt}\pmb{\varphi}_t,
\end{multline}
where $c_{mk} \overset{\triangle}{=} \frac{\sqrt{\tau_p\rho_p}\beta_{mk}}{\tau_p\rho_p\sum\nolimits_{j=1}^{K}\sum\nolimits_{t'=1}^{\tau_p}\sum\nolimits_{t=1}^{\tau_p}\beta_{mj}|p_{jt'}\bm{\varphi}^H_{t'}p_{kt}\bm{\varphi}_{t}|+1}$. The second moment of the channel estimate $\hat{\mathbf{g}}_{mk}$ is expressed in closed form as $\gamma_{mk} \overset{\triangle}{=} \mathbb{E}\left\{|\hat{\mathbf{g}}_{mk}|^2\right\} = \sqrt{\tau_p\rho_p}\beta_{mk}c_{mk}$.

\subsubsection{Downlink Data Transmission}
The APs apply the classic conjugate beamforming technique for the transmit precoding (TPC) of signals by exploiting the channel estimates. In particular, the signal transmitted from AP~$m$ to user~$k$ is $\mathbf{x}_m = \sqrt{\rho_d}\sum\nolimits_{k =1}^K \sqrt{\eta_{mk}}  \hat{\mathbf{g}}_{mk} q_k$, where $\rho_d$ is the maximum normalized transmit power at each AP;  $q_k$ is the symbol intended for user~$k$ with $\mathbb{E}\{ |q_k|^2 \} =1$; and $\eta_{mk}$ is the control coefficient satisfying the power budget constraint of $L\sum\nolimits_{k=1}^K \eta_{mk}\gamma_{mk}\leq 1$. The signal received by user~$k$ is 
\begin{align}
    \label{eq:received-signal-UE9}
    r_k &= \sum\nolimits_{m=1}^M\mathbf{g}_{mk}^H \mathbf{x}_m + w_{k} \nonumber \\
    &= \sqrt{\rho_d}\sum\nolimits_{m=1}^M\sum\nolimits_{j=1}^K\sqrt{\eta_{mk}}\mathbf{g}_{mk}\hat{\mathbf{g}}_{mj}^*q_j+w_k,
\end{align}
where $w_k$ is the AWGN at user~$k$ distributed as $\mathcal{C}\mathcal{N}(0,1)$. From \eqref{eq:received-signal-UE9}, one can obtain the achievable downlink ergodic throughput of user~$k$ as formulated in Lemma~\ref{lemma:Ergodic}.
\begin{lemma} \label{lemma:Ergodic}
 The ergodic throughput of user~$k$ is expressed in closed form in \eqref{eq:downlink rate} with $B$ being the system bandwidth.
\end{lemma}
\begin{proof}
First, a lower bound on the downlink channel capacity is attained by assuming that the APs apply linear precoding and then forget the TPC vectors when computing the throughput. Then, the closed-form expression is obtained by formulating the moments of Gaussian random variables.  
\end{proof}
The ergodic throughput expressed in \eqref{eq:downlink rate} only depends on the channel statistics so that the formulated may rely upon a specific pilot assignment policy for a long period of time.
\begin{figure*}[t]
 \fontsize{9}{9}{   \begin{equation}
        \label{eq:downlink rate}
        \begin{split}
    R_{dk} = B\text{log}_2\left ( 1+ \frac{\rho_d L^2\left(\sum\limits_{m=1}^M\sqrt{\eta_{mk}}\gamma_{mk}\right)^2}{\rho_dL^2\sum\limits_{j=1, j\neq k}^K\left (\sum\limits_{m=1}^M\sqrt{\eta_{mj}}\gamma_{mj}\frac{\beta_{mk}}{\beta_{mj}}\right)^2\left|\sum\limits_{t=1}^{\tau_p}\sum\limits_{t'=1}^{\tau_p}p_{jt}\bm{\varphi}^H_tp_{kt'}\bm{\varphi}_{t'}\right| + \rho_dL\sum_{j=1}^K\sum\limits_{m\in \mathcal{M}(k)}^M \eta_{mj}\gamma_{mj}\beta_{mk}+1} \right )
        \end{split}
    \end{equation}}
    \hrulefill
    \vspace{-2mm}
\end{figure*}

\subsection{Problem Formulation}
We aim for maximizing the total ergodic rate by optimizing the pilot assignment, which is mathematically formulated as

\begin{subequations}\label{opt:problem} 
    \begin{align}
        \underset{\{p_{kt}\} }{\text{maximize}} \quad &\sum\nolimits_{k=1}^K R_{dk} \label{opt:optimization} \\ 
        \text{subject to} \quad &\sum\nolimits_{t = 1}^{\tau_p} p_{kt} = 1, \forall k, \label{opt:constraint1} \\ 
        &p_{kt} \in \{0,1\}, \forall k, t, \label{opt:constraint2} 
    \end{align}
\end{subequations}
where the constraint \eqref{opt:constraint1} ensures that each user only has a single pilot signal. The constraint \eqref{opt:constraint2} makes Problem~\eqref{opt:problem} combinatorial. It is computationally challenging to obtain the optimal solution as the system serves many users. Even though an exhaustive search would indeed find the globally optimal solution, it is excessively complex for large-scale systems. By contrast, heuristic methods are capable of reducing the complexity at the cost of performance reduction. Against this backdrop, we design a hybrid quantum neural network to solve Problem~\eqref{opt:problem} by maximizing the total ergodic throughput.

\section{Proposed Model}
\label{sec:HQCNN}
This section presents the proposed HQCNN consisting of three main parts: pre-processing, quantum convolutional neural network (QCNN), and post-processing as illustrated in Fig. \ref{fig:HQCNN-sys}.

\begin{figure*}[!t]
        \centering
\includegraphics[width=0.7\linewidth]{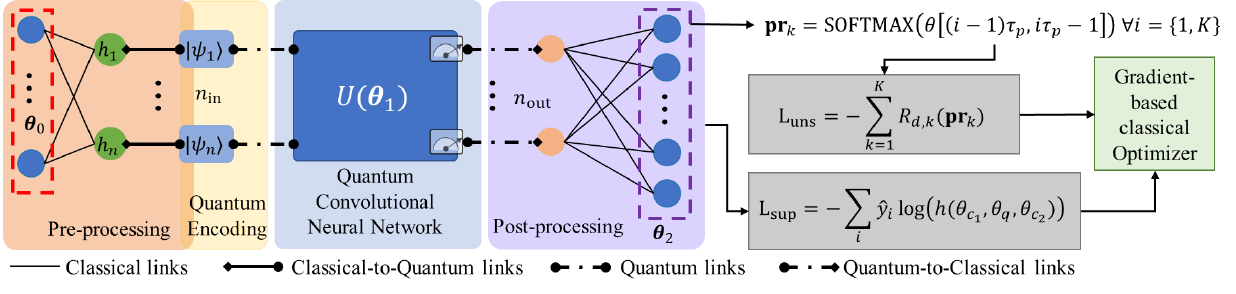}
        \caption{Architecture of the proposed HQCNN model: the pre-processing layer for embedding classical data to quantum state; QCNN for processing quantum states; the post-processing layer for converting outcomes of QCNN. }
        \label{fig:HQCNN-sys}
        \vspace{-3mm}
\end{figure*}

\subsection{Design of HQCNN Model}
\label{sec: Design}
\subsubsection{Pre-processing layer}
\label{sec: pre-processing layer}
To utilize the QML, we must transform the classical data into the quantum state in a higher-dimensional Hilbert space. A fully connected (FC) and a quantum embedding layer (QEL) are first deployed. Therein, the linear layer adjusts the dimension of classical input data to the $n_{\mathrm{in}}$-dimensional aggregated vectors, where $n_{\mathrm{in}}$ is the predefined number of qubits in the quantum circuits. \textcolor{black}{In this work, we adopt the angle embedding method over basis and amplitude embedding techniques. The rationale is that basis embedding is limited to binary inputs, while amplitude embedding demands deeper quantum circuits, thereby increasing complexity and vulnerability to quantum noise \cite{munikote2024comparing}. In contrast, angle embedding produces shorter output vectors, which helps reduce the number of trainable parameters in the pre-processing layer. These advantages contribute to lower circuit complexity and a faster convergence rate during training \cite{munikote2024comparing}.} Note that the initial quantum state $|\psi\rangle$ is represented as
$|\psi\rangle = \bigotimes_{i=1}^{n_{\text{in}}} \left(\cos{(h_i)}|0\rangle + \sin{(h_i)}|1\rangle\right),$
where $h_i$ is the $i$-th element of $\mathbf{h} = \mathbf{\theta_1}^T\mathbf{x}+\mathbf{b}$ in which $\mathbf{x}$ is the classical input vector; $\mathbf{\theta_1}$, and $\mathbf{b}$ are weighting matrix and bias vector of the linear layer, respectively. 

\subsubsection{Quantum convolution layer}
\label{sec: QCL}
A QCNN containts parameterized quantum circuits, each similar to a filter which process a quantum state as its input and outputs a vector of expectation values. As reported in \cite{hur2022quantum}, a QCNN is capable of adopting the core principles of convolutional neural networks, including the convolutional and pooling layers for feature mapping and dimension reduction, as shown in Fig.~\ref{fig:CNN-to-QCNN}a. By harnessing this, a kernel is slid over the input data, performing element-wise operations in order to transform data from the input to the output in CNN. By exploiting parameter-sharing across the layers, the proposed HQCNN is capable of significantly reducing the number of training parameters with the aid of a series of quantum convolution layers (QCLs) and quantum pooling layers (QPLs). Fig.~\ref{fig:CNN-to-QCNN}b illustrates example of one QCNN layer. The QCLs employ PQCs to aggregate quantum information from the adjacent qubits, while the QPLs reduce the network size by measuring half of the qubits. These layers are repeated until the main features of information are successfully extracted.
\begin{figure}[t]
    \centering
    \includegraphics[width=0.8\columnwidth]{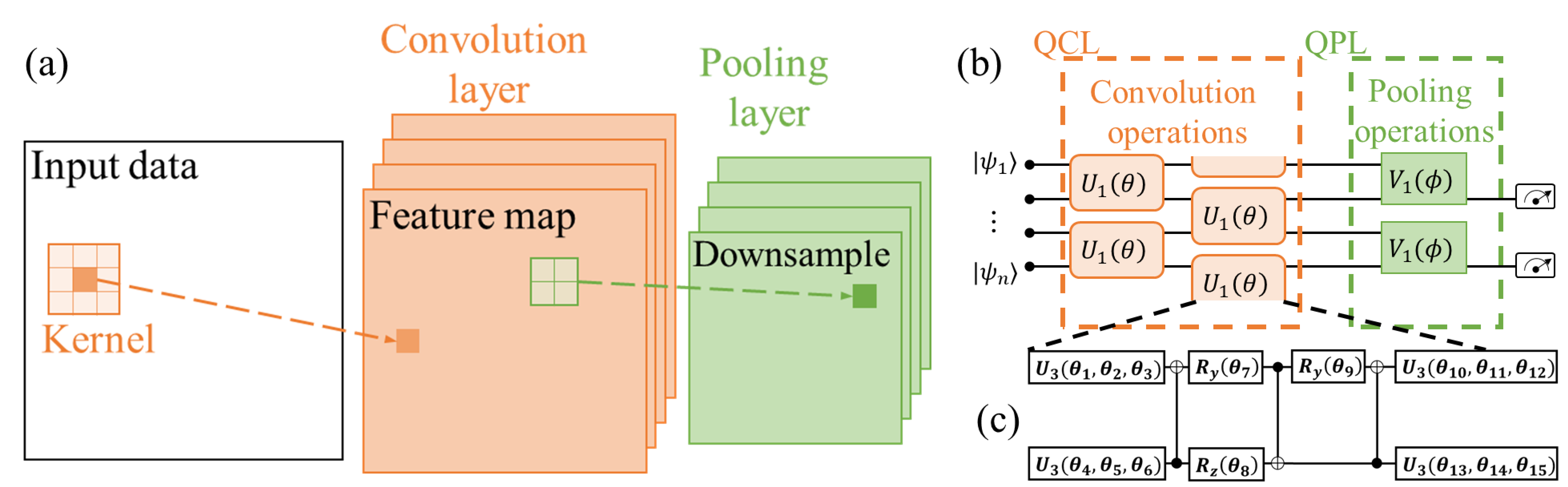}
    \caption{{(a) Illustration of a CNN layer; (b) QCNN layer; and (c) PQC design.}}
    \label{fig:CNN-to-QCNN}
    \vspace{-4mm}
\end{figure}

Again, in the previous QCNN design of \cite{hur2022quantum}, each QCL exploited identical PQC structures but along with different parameter sets. This approach required a total of $T\sum_{i=1}^{\log_2(n)}2^i$ parameters if only a single outcome is measured, where $T$ is the number of parameters per PQC. We therefore propose a novel QCL design inspired by the concept of CNN kernel parameter-sharing \cite{kim2020deep} for reducing the number of learnable parameters. Using the same set of parameters for all the PQCs at each layer reduces the total number of parameters to $T\log_2(n_0)$, where $n_0$ is the initial number of qubits. The exponential parameter reduction capability of our design significantly reduces the computational cost, making the quantum system more efficient. The PQC structure of our neural network design is adopted from the structure of \cite{hur2022quantum}, relying on the detailed illustration seen in Fig.~\ref{fig:CNN-to-QCNN}c.

\vspace{-1mm}
\subsubsection{Post-processing layer}
\label{sec: post-processing layer}

Finally, a fully connected layer is used at the output for processing the collapsed classical outcomes of the QCNN. The post-processing layer takes the $n^{\text{out}}$-outcome expectation values of the QCNN as its input. We view the pilot assignment problem in the scope of a multi-class classification problem with the output of a $(K \times \tau_p)$ matrix.

\subsection{Training Procedure}
\label{sec: Training}
\subsubsection{Inference process}

The proposed HQCNN model performs the pilot assignment by only considering the LSF coefficients hosted by the matrix $\bm{\beta} \in \mathbb{C}^{M \times K}$, since the achievable data rate of each user predominantly relies on its LSF with respect to other APs. The input layer of the proposed hybrid neural network maps the matrix $\bm{\beta}$ to the initial quantum state. In QCNN layer, the output state of the $l$-th layer is
\begin{align}
    |\psi_i(\bm{\theta}_i)\rangle\langle\psi_i(\bm{\theta}_i)| = \text{Tr}_{O_i}\left(U_i(\bm{\theta}_i)|\psi_{i-1}\rangle\langle\psi_{i-1}|U_i(\bm{\theta}_i)^{\dagger}\right),
\end{align}
where $\text{Tr}_{O_i}(\cdot)$ represents the partial trace operation over the subsystem $O_i$; $U_i (\cdot)$ is the parameterized unitary operation combining the quantum convolution and pooling operations; and $\bm{\theta}_i$ denotes the parameters of the PQC in the $i$-th quantum convolutional layer. After processing by all the QCNN layers, the output is calculated as in \cite{zhou2023towards}, given by
\begin{align}
    \label{eq: measurement}
    \langle\mathcal{M}_i\rangle = \langle\bm{\psi}|U(\theta)^{\dagger}A_iU(\theta)|\bm{\psi}\rangle,
\end{align} 
where $|\bm{\psi}\rangle$ is the final quantum state, $U(\bm{\theta})$ is the overall unitary operation composed of all the quantum layers, and $O_i$ are the observables, typically chosen as the Pauli operators. In this study, we use the Pauli-Z operators for the measurement layer. The final outcomes are estimated as expectation values $\langle\mathcal{M}_i\rangle$ to avoid probabilistic measure noise, serving as the input of the classical post-processing layer. In the post-processing layer, we apply a softmax function to each row of the $(K \times \tau_p)$ output matrix to compute the pilot selection probabilities. 

For supervised learning, the pilot assignment orders obtained from \cite{zaher2022learning} are exploited as labels. Users select the specific pilot signals with the highest probability. In this framework, the cross-entropy loss function is employed  as
\begin{align}
    \label{eq: LF-supvervised}
    \text{LF}_{\text{sup}} = -\left\| \sum\nolimits_{k=1}^K\sum\nolimits_{i=1}^{\tau_p} \mathbf{y}_{ki}\log(h(\bm{\beta};\bm{\theta})_{ki})\right\|_1,
\end{align}
where $h(\bm{\beta};\bm{\theta})$ is the prediction probability and $\bm{\theta}$ represents the trainable parameters. Since both the gradient-based optimizers such as the classical backpropagation and quantum parameter shift rule \cite{hanzo2025quantum} are based on the gradient descent method, both classical and quantum machine learning models minimize the loss function iteratively. Hence, if the unsupervised loss function is reformulated as the negative version of the sum-rate function influenced by the pilot probabilities, as in (\ref{eq: LF-unsupervised}), where $\mathbf{q}_k \in \mathbb{R}^{\tau_p}$ denotes the probability of the pilot selection for user $k$, then minimizing the loss function process is similar to maximizing the total throughput function, in line with the objective of Problem (\ref{opt:problem}).

\begin{figure*}[!t]
    \normalsize
    \begin{equation}
        \label{eq: LF-unsupervised}
        \begin{split}
    \text{LF}_{\text{uns}}^{(k)} = -B\sum\limits_{k=1}^K\mathbb{E}\left\{\log_2\left ( 1+ \frac{\rho_d L^2\left(\sum\nolimits_{m=1}^M\sqrt{\eta_{mk}}\gamma_{mk}\right)^2}{\rho_dL^2\sum\nolimits_{j=1,j\neq k}\left (\sum\nolimits_{m=1}^M\sqrt{\eta_{mj}}\gamma_{mj}\frac{\beta_{mk}}{\beta_{mj}}\mathbf{q}_{k}^H\mathbf{q}_j\right)^2 + \rho_dL\sum\nolimits_{j=1}^K\sum\nolimits_{m=1}^M \eta_{mj}\gamma_{mj}\beta_{mk}+1} \right )\right\}.
        \end{split}
    \end{equation}
\hrule
\vspace{-4mm}
\end{figure*}

\subsubsection{Updating process}

In the CNN model, the backpropagation method updates the weights in kernel by computing gradients of the loss function with respect to the parameters and adjusting them using the gradient descent method. Unlike the CNN model having visible activations in hidden layers, the quantum states in QCNNs are not measurable until they collapse, preventing QCNNs from computing gradients via the chain rule. Hence, the QCNNs utilize another gradient-based approach, namely the parameter shift method of \cite{hanzo2025quantum} to estimate the natural quantum gradient, given as 
\begin{align}
        \frac{\partial\langle \mathcal{M} \rangle(\theta)}{\partial\theta} =  \frac{\langle \mathcal{M} \rangle\left(\theta+s\right)+\langle \mathcal{M} \rangle\left(\theta-s\right)}{2\sin(s)},
\end{align}
where $s$ is the parameter shift value. For an arbitrary quantum state $|\psi\rangle = a|0\rangle + b|1\rangle$, the outcome of a QCNN is expressed by a single-qubit unitary matrix as
\begin{equation}
        \begin{pmatrix}
e^{-i\alpha} \cos(\varphi) & -e^{i\beta} \sin(\varphi) \\
e^{-i\beta} \sin(\varphi) & e^{i\alpha} \cos(\varphi)
\end{pmatrix},
\end{equation}
where $\alpha$ is the phase, while $\beta$ and $\varphi$ define the angle of general rotation. Due to the periodic nature of rotation functions, the shift value is selected as $(\pi/2)$ for simplicity. Hence, the gradient is formulated as
\begin{align}
        \frac{\partial\langle \mathcal{M} \rangle(\theta)}{\partial\theta} = \frac{1}{2} \left( \langle \mathcal{M} \rangle\left(\theta+\frac{\pi}{2}\right)+\langle \mathcal{M} \rangle\left(\theta-\frac{\pi}{2}\right)\right).
\end{align}
After getting the gradients, new parameters are calculated similarly to the popular gradient-descent methods, as
\begin{align}
    \theta_{\text{new}}=\theta_{\text{old}}-\epsilon\frac{\partial\langle \mathcal{M} \rangle(\theta)}{\partial\theta}.
\end{align}
Instead of requiring backpropagation, the quantum circuit has to be estimated twice along with $\theta\pm\pi/2$. Hence, by significantly reducing the number of trainable parameters, the computational complexity in our design of the updating process is exponentially lower than those of the benchmark design in \cite{hur2022quantum}. Based on, the construction and training process conceived, the proposed HQCNN can perform the pilot assignment, as formulated in Lemma~\ref{Lemma3}. 
\begin{lemma} \label{Lemma3}
 The proposed HQCNN model relying on $n\in\mathbb{N}$ number of qubits in the PQC, and denoted as $\boldsymbol{\theta}^*$, is capable of solving Problem~\eqref{opt:optimization} for cell-free massive MIMO systems configured by the triplet $\{M , K, \tau_p \}$ via approximating the mapping from the large-scale fading coefficients to the pilot selection probabilities.
 \vspace{-2mm}
\end{lemma}
\begin{proof}
    The proof is available in the appendix. 
\end{proof}
\vspace{-6mm}

\section{Numerical Results}
\label{sec: Result}

Let us now evaluate the performance of the proposed HQCNN under both supervised and unsupervised training frameworks. For a data-driven scenario, the proposed HQCNN model is compared to both the MLP as well as to the light- and heavy-CNN models, which are inspired by the CNN of \cite{kim2020deep}, regarding the convergence rate of the supervised learning and computational efficiency of the unsupervised framework. Table~\ref{tab:parameter-comparision} analyses the complexity order of each model. While the architecture of both our model and of the MLP rely on $M$, $K$, $\tau_p$ and $n$, the structure of CNN depends on other factors like the kernel size (in pixels) $k^2$, the number of convolutional layers $F$, and the number of input and output channels in each convolution layer $C_{\text{in}}$, $C_{\text{out}}$. Both CNNs use $3\times3$ kernels in two convolutional layers, but the number of output channels are $(8, 16)$ and $(32, 64)$ in light-CNN and heavy-CNN, respectively. In this work,  $8$ qubits are utilized in the quantum model, the number of trainable parameters in HQCNN approximately equal to those of MLP and light-CNN, while much lower than those of heavy-CNN, as a benefit of quantum parameter-sharing harnessed for reducing the computational complexity. The model-based approaches used for our comparison involve the random (RPA), Greedy (GPA) \cite{ngo2017cell},  Master-AP (MPA) \cite{zaher2022learning}, location-based (LPA) \cite{nguyen2023efficient} and exhaustive search (EPAS).

\begin{table}[t]
    \centering
    \caption{Parameters analysis.}
    {\footnotesize
    \resizebox{0.9\columnwidth}{!}{
    \begin{tabular}{|l|l|}
        \toprule
         \textbf{Model} & \textbf{Number of parameters}    \\
         \midrule
         HQCNN          &   $O\left(nK(M+\tau_p) +15 \text{ (quantum parameters)}\right)$\\
         \midrule
         HQCNN \cite{hur2022quantum} & $O\left(nK(M+\tau_p) +15n \text{ (quantum parameters)}\right)$ \\
         \midrule
         CNN    &   $O\left(k^2\left(\sum_{i=1}^FC_{\text{in},i}C_{\text{out},i}+M\tau_pK^2\right)\right)$ \\
         \midrule
         MLP                &   $O\left(nK(M+\tau_p)+n^2\right)$ \\
         \bottomrule

    \end{tabular}
    }
    }
    \label{tab:parameter-comparision}
    \vspace{-5mm}
\end{table}  

We consider a one-square-kilometer area having the system configurations as follows: pilot transmit power $ 100  \mbox{ [mW]}$; access point downlink transmit power $200\mbox{ [mW]}$; bandwidth $20 \mbox{ [MHz]}$; carrier frequency $f = 1.9  \mbox{ [GHz]}$; noise figure of $9  \mbox{ [dB]}$. The LSF coefficient $\beta_{mk}$ models the path-loss and shadow fading as follows $\beta_{mk} = \text{PL}_{mk}\cdot10^{\frac{\sigma_{sh}z_{mk}}{10}}$, where $\text{PL}_{mk}$ represents the path-loss, and $10^{\frac{\sigma_{sh}z_{mk}}{10}}$ represents the shadow fading having the standard deviation $\sigma_{sh}$, and $z_{mk}\sim\mathcal{N}(0,1)$. The path-loss is represented by the three-slope model as
\small
\begin{align*}
        \text{PL}_{mk} \!=\!
    \begin{cases}
    -L\! - \!35 \log_{10}(d_{mk}), \!& \text{if } \!d_{mk} > d_1 \\
    -L \!- \!15 \log_{10}(d_1) \!- \!20 \log_{10}(d_{mk}),\! & \text{if } \!d_0 < d_{mk} \leq d_1 \\
    -L \!- \!15 \log_{10}(d_1) \!- \!20 \log_{10}(d_0), \!& \text{if } \!d_{mk} \leq d_0.
\end{cases}
\end{align*}
\normalsize
This setting was used in \cite{ngo2017cell}, along with similar selections of $L = 141 $, $d_1 = 50$m, $d_0=10$m and $\sigma_{sh}=8$dB. The power control coefficients are uniformly set as $\eta_{mk} = 1/L\sum_{k=1}^K\gamma_{mk}$ for all the APs.

\begin{figure*}[t]
    \centering
    \subfloat[]{
    \includegraphics[width=0.3\textwidth]{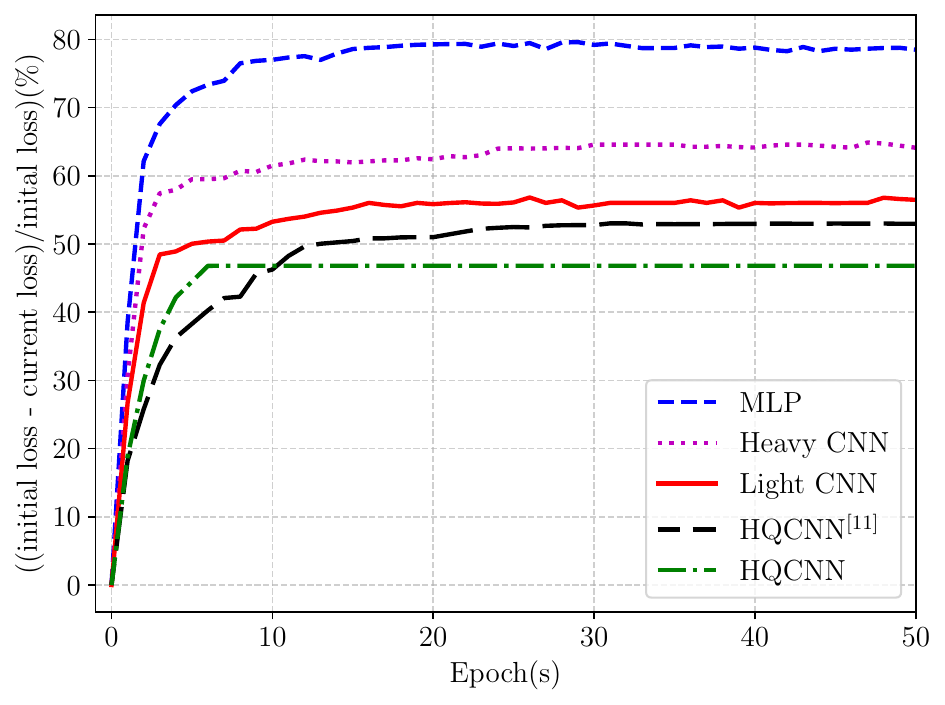} 
        \label{fig:sup-loss}
    }
    \hfill
    \subfloat[]{
    \includegraphics[width=0.3\textwidth]{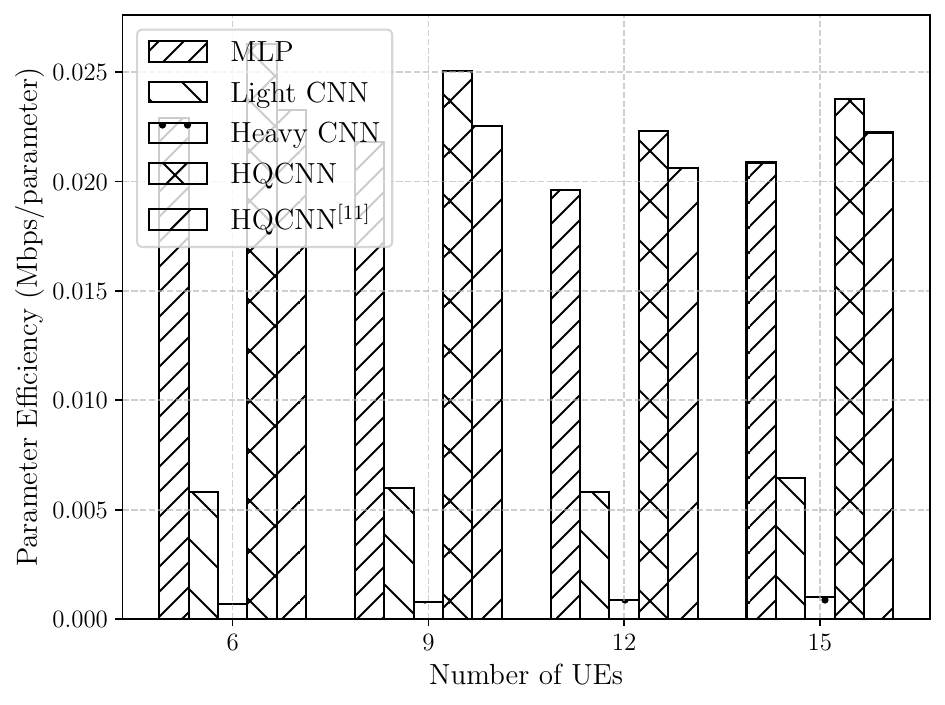} 
        \label{fig:param_efficiency}
    }
    \hfill
    \subfloat[]{
    \includegraphics[width=0.3\textwidth]{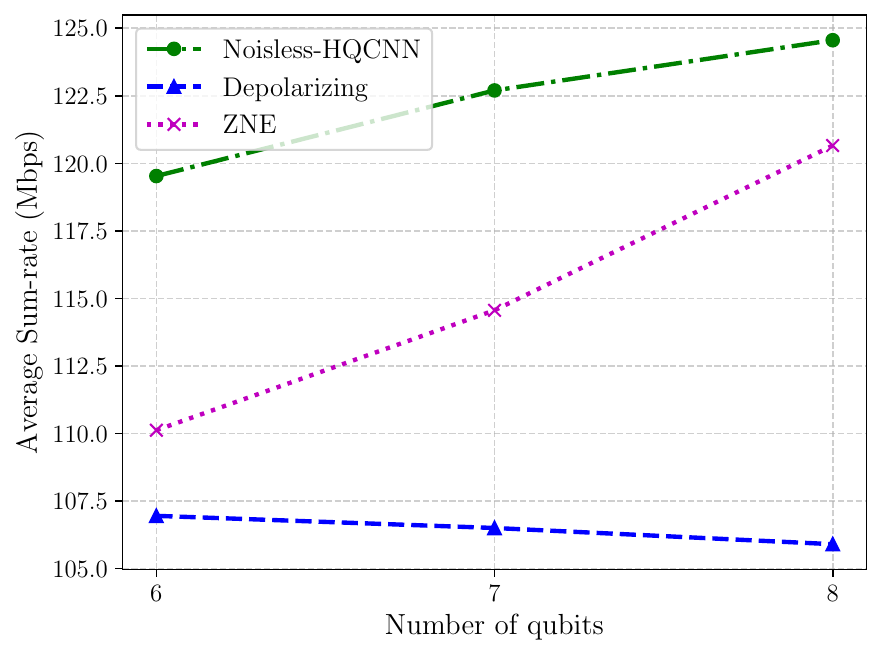} 
        \label{fig:noise}
    }
    \caption{(a) Proportion of training loss reduction relative to the initial loss value versus the training epochs (supervised);  (b) Parameter efficiency of learning-based benchmarks (unsupervised); and (c) the total ergodic throughput with $(M,L,K,\tau_p) = (40,2,20,10)$ in noiseless- and noise-simulation (unsupervised).}
    \label{fig:result}
    \vspace{-0.35cm}
\end{figure*}

\subsection{Performance of Supervised Learning} 
We select the \textit{Master-AP} method as the label generator algorithm and use the cross-entropy loss function of \eqref{eq: LF-supvervised}. The models are trained iteratively to minimize the gap between the predictions and the ground truth labels. Fig. \ref{fig:sup-loss}, which displays the proportion of training loss reduction relative to the initial loss value versus the training epochs. Both the DNN-based models and conventional HQCNN \cite{hur2022quantum} exhibit slower convergence due to their increased computational overhead and more complex optimization landscapes, requiring more than $20$ epochs to converge. By contrast, the proposed HQCNN converges significantly faster, stabilizing within approximately $7$ epochs, showing the advantage of our design. Moreover, according to the Google Quantum AI document \cite{GoogleQuantumAI}, the estimated quantum simulation time of a 16-qubit model is $0.009$ seconds and $0.023$ seconds for a classical c2-standard-60 CPU processor in case of noiseless and noisy scenarios, respectively. This shows that QML has the potential to speed up the computational process even at the simulation level. Note that the processing speed of the quantum system significantly depends on the hardware configuration.

\subsection{Performance of Unsupervised Learning} 

In Table~\ref{tab:EPAS-comparison}, we compare the performance of the learning-based benchmarks and EPAS, scaling up the system from small to medium size. In a small system, the EPAS provides the optimal solution at the highest cost, which becomes excessively complex if $K > 15$. Both our proposed HQCNN and the HQCNN of \cite{hur2022quantum} achieve approximately   $98\%$ of the global optimum, becoming the second-best benchmark. However, the hybrid quantum models still remain effective as the system-scale is expanding, demonstrating robustness in scenarios where the EPAS fails to find the solution, clearly outperforming other methods in all tested scenarios. The DNN-based models' performance is consistently lower than that of ours, with approximately $5\%$, $10\%$, and $13\%$ lower for heavy-CNN, light-CNN, and MLP, respectively. We also observe that heuristic methods such as RPA, GPA, MPA, and LPA, consistently yield lower average sum-rates than the proposed HQCNN model. The former becomes even more inferior when the system scales up. For example, in small-scale settings of $(M,L,K,\tau_p) = (10,2,6,3)$, the best-performing heuristic method (MPA) achieves $17.21$ (Mbps), $4.8\%$ lower than our model $(18.08\%)$. As the system scales to medium sized configurations (e.g., $M=30\sim45$), the limitations of heuristic methods become more pronounced. Quantitatively, the exhibit performance gaps up to $20.4\%$. For instance, at $M=45$, $K=20$, LPA yields $111.05$ Mbps, compared to our HQCNN's $133.75$ Mbps, while GPA and MPA are outperformed by about $17\%$. These methods rely on localized optimization criteria, such as minimizing interference for individual UEs, specific APs, or within spatially bounded regions. These strategies falter in complex multi-user environments. For example, GPA prioritizes low-performance users and overlooks high-rate UEs; MPA restricts decisions to the single strongest AP per user; and LPA fails when UEs are located densely in specific areas. These observations have shown the effectiveness of our model in attaining superior sum-rate over other methods.

\begin{table*}[t]
    \centering
    \caption{Average sum-rate (Mbps) comparison in small-scale and medium-scale cell-free massive MIMO systems (unsupervised).}
    \renewcommand{\arraystretch}{1.2} 
    \resizebox{0.8\textwidth}{!}{ 
    \begin{tabular}{|c|ccc|cccccccccc|}
        \toprule
        \textbf{System size}    & $M,L$     & $K$   & $\tau_p$  & \textbf{Random} & \textbf{Greedy \cite{ngo2017cell}}   & \textbf{Master-AP \cite{zaher2022learning}}  & \textbf{Location-based \cite{nguyen2023efficient}}   & \textbf{MLP}  & \textbf{\makecell{Light-CNN}} & \textbf{\makecell{Heavy-CNN}} & \textbf{HQCNN \cite{hur2022quantum}}  & \textbf{Ours} &  \textbf{EPAS} \\
        \midrule
Small  & $(10,2)$ & $6$  & $3$ & $15.75$ & $15.76$ & $17.21$ & $16.50$ & $17.14$ & $17.20$ & $17.51$ & $18.05$ & $ \underline{18.08}$   & $\bm{18.43}$  \\
       & $(10,2)$ & $9$  & $3$ & $18.34$ & $19.07$ & $22.13$ & $20.67$ & $21.82$ & $22.93$ & $23.13$ & $23.80$ & $ \underline{23.84}$   & $\bm{24.29}$  \\
       & $(10,2)$ & $12$ & $3$ & $23.20$ & $24.76$ & $26.22$ & $25.54$ & $25.73$ & $27.18$ & $27.79$ & $28.19$ & $ \underline{28.19}$   & $\bm{28.74}$  \\
       & $(10,2)$ & $15$ & $3$ & $27.18$ & $33.68$ & $35.03$ & $34.55$ & $33.89$ & $35.71$ & $36.85$ & $37.36$ & $ \underline{37.45}$   & --            \\
       \midrule
Medium  & $(30,2)$ & $20$  & $10$ & $67.67$ & $82.93 $  & $78.93 $ & $82.16 $ & $77.65 $ & $89.77 $ & $93.30 $ & $99.41$ & $ \underline{99.51 }$   & --  \\
        & $(35,2)$ & $20$  & $10$ & $78.34$ & $96.98 $  & $90.98 $ & $94.45 $ & $89.69 $ & $103.34$ & $109.95$ & $115.44$ & $\underline{115.44}$   & --  \\
        & $(40,2)$ & $20$  & $10$ & $88.43$ & $102.31$  & $98.20 $ & $100.12$ & $96.51 $ & $111.17$ & $118.79$ & $123.72$ & $\underline{124.56}$   & --  \\
        & $(45,2)$ & $20$  & $10$ & $91.22$ & $112.45$  & $110.45$ & $111.05$ & $109.60$ & $122.31$ & $130.20$ & $133.04$ & $\underline{133.75}$   & --  \\
        \bottomrule
    \end{tabular}%
    }
    \label{tab:EPAS-comparison}
    \vspace{-4mm}
\end{table*}

To characterized the performance versus complexity, we consider an unconventional efficiency metric, the ratio of the sum-rate and the number of parameters as seen in Fig. \ref{fig:param_efficiency}. Although, the performance of our model and the HQCNN in \cite{hur2022quantum} are equal, our model obtains higher per-parameter-throughput since our design utilizes less quantum trainable parameters. However, the heavy CNN can obtain higher total throughput than the light CNN and MLP, it requires much more hardware resources, leading to the lowest parameter efficiency. These results suggest that our model may indeed be deemed efficient over benchmarks. 

\vspace{-2mm}
\subsection{Limitation}

Despite its potential advantages, HQCNN remains constrained by the limitations of the near-term quantum hardware. The model is particularly sensitive to noise sources such as gate noise, shot noise, and decoherence, which accumulate and degrade the performance during training. Depolarizing noise, arising from imperfect gate operations, leads to information loss, and shot noise reflects statistical fluctuations, owing to using a finite number of measurement samples. While shot noise can be reduced by increasing the number of measurement shots, trading between operation time and the accuracy of circuits, depolarizing can be mitigated by zero-noise extrapolation (ZNE) techniques, which run the circuit at different noise levels and extrapolate back to the zero-noise limit \cite{larose2022mitiq}. Decoherence, caused by environmental interactions, results in the erosion of quantum behavior. This can be addressed by adopting low-depth PQCs to keep circuit execution time within coherence limits \cite{du2025quantum}. Accordingly, HQCNN utilizes depth-6 PQCs to minimize decoherence risk.

The impact of noise at different number of qubits is illustrated in Fig. 3c for $(M,L,K,\tau_p) = (40,2,20,10)$. The proposed model falters in noisy environments, at a depolarization error rate of $0.1$, especially upon increasing the number of qubits. However, the ZNE technique can significantly mitigate depolarizing errors. These observations indicate that increasing the number of qubits can theoretically enhance the computational power of the model by providing richer feature maps. However in practical noisy systems, ZNE might inflict additional errors, hence degrading the performance.

\section{Conclusion}

A powerful hybrid quantum CNN model was conceived to assign the pilot signals in cell-free massive MIMO systems by maximizing the total ergodic throughput. The proposed HQCNN utilizes parameterized quantum circuits for efficient feature extraction, leveraging quantum-domain advantages, including superposition and entanglement. Our hybrid model significantly reduces trained parameters, so resulting in faster convergence and low cost. Numerical results showed that HQCNN obtained a near-optimal ergodic throughput for small-scale systems and  competitive performance in large-scale systems, highlighting its efficiency in pilot assignment tasks. 

\vspace{-1.5mm}

\appendix
\section{Proof of Lemma \ref{Lemma3}}
\label{appendix:A}
Let us view the pilot assignment in Problem~\eqref{opt:problem} as a mapping. The proposed HQCNN model should be interpreted as: 1) the MLP for pre-processing $f_0(\pmb{\beta};\boldsymbol{\theta}_0): \mathbb{R}^{M\times K} \rightarrow \mathbb{R}^{n_{\mathrm{in}}}$; 2) the QCNN delighted as, $f_1(\mathbf{z}_0;\pmb{\theta}_1): \mathbb{R}^{n_{\mathrm{in}}} \rightarrow \mathbb{R}^{n_{\mathrm{out}}}$, where $\mathbf{z}_0 = f_0(\pmb{\beta};\pmb{\theta}_0)$; and 3) the MLP for post-processing $f_2(\mathbf{z}_1;\pmb{\theta}_2): \mathbb{R}^{n_{\mathrm{out}}} \rightarrow \mathbb{R}^{K\times \tau_p}$, where $\mathbf{z}_1 = f_1(\mathbf{z}_0;\pmb{\theta}_1)$. The overall mapping of the proposed HQCNN model is 
   \begin{equation}
        f(\pmb{\beta};\boldsymbol{\theta}) \!= \!f_2(f_1[f_0(\pmb{\beta};\pmb{\theta}_0);\pmb{\theta}_1];\pmb{\theta}_2): \!\mathbb{R}^{M\times K} \rightarrow \mathbb{R}^{K\times \tau_p}.
    \end{equation}
According to the universal approximation theorems of classical ML \cite{HORNIK1989359},  quantum ML \cite{gonon2023universalapproximationtheoremerror}, and  combining the three main mappings, the  error bound of $f(\boldsymbol{\beta};\boldsymbol{\theta})$ is formulated as
\begin{equation} \label{eq:Bound}
    \begin{aligned}
            &\|f(\pmb{\beta};\boldsymbol{\theta})-g^\ast({\pmb{\beta}})\|_F = \! \|f_2(f_1(f_0(\boldsymbol{\beta};\boldsymbol{\theta}_0);\boldsymbol{\theta}_1);\boldsymbol{\theta}_2) - g^\ast({\boldsymbol{\beta}})\|_F\\
            &\stackrel{(a)}{\leq}   \|f_0(\boldsymbol{\beta};\boldsymbol{\theta}_0)-g_1^*({\boldsymbol{\beta}})\|_F + w_0\|f_1(\boldsymbol{z}_0;\boldsymbol{\theta}_1)-g_2^*(\boldsymbol{z}_0)\|_F \\
            & + w_1\|f_2(\boldsymbol{z}_1;\boldsymbol{\theta}_2)-g_3^*(\boldsymbol{z}_1)\|_F 
            \stackrel{(b)}{\leq} \epsilon_{0} + w_0\epsilon_{1}n^{-0.5} + w_1\epsilon_{2}, 
    \end{aligned}
\end{equation}
     where $g^\ast(\pmb{\beta})$ is the ideal mapping  of the large-scale fading coefficients $\pmb{\beta}$; $\epsilon_0$ and $\epsilon_2$ are the tolerable errors of the MLP; $\epsilon_1$ is the quantum approximation error depending both on the number of qubits and on the design of the PQC; $w_0$ and $w_1$ represent the relative error contributions, which demonstrate that both $f_1(\cdot;\cdot)$ and $f_2(\cdot; \cdot)$ are affected by the output of the previous blocks \cite{saltelli2008global}. In \eqref{eq:Bound}, $(a)$ is obtained by exploiting the sequential nature of $f(\pmb{\beta};\pmb{\theta})$ as a composition of the functions that the summation of errors is bounded across each stage \cite[Chapter~5]{shalev2014understanding} and $(b)$ is derived by the universal approximation theorem. Since $w_0, w_1 \rightarrow 1$ as $\epsilon_{0}, \epsilon_{1} \rightarrow 0$,  the total error can be made arbitrarily small by choosing a sufficiently large number of qubits $n$ and appropriately designing the classical models. Thus, an HQCNN can find the ideal mapping function between the large-scale coefficients and the pilot assignment strategy, which completes the proof.


\begingroup
\footnotesize
\bibliographystyle{IEEEtran}
\bibliography{reference-letter}
\endgroup

\end{document}